\documentclass[preprint,12pt,authoryear]{elsarticle}
\usepackage{graphicx}
\usepackage{amsmath,amssymb}
\usepackage{multirow}
\usepackage{amsbsy}
\usepackage{color}
\usepackage{multirow}
\usepackage{amsthm} 
\usepackage[a4paper,width=165mm,top=25mm,bottom=25mm]{geometry}
\newtheorem{theorem}{Theorem}[section]
\newtheorem{lemma}{Lemma}[section]
\newtheorem{corollary}{Corollary}[section]
\newtheorem{remark}{Remark}[section]
\newtheorem{exmp}{Example}[section]

\numberwithin{equation}{section}
\numberwithin{table}{section}
\everymath=\expandafter{\the\everymath\displaystyle}

\begin{document}

\begin{frontmatter}

\title{Optimal two-level choice designs for the main effects \\ and specified interaction effects model}

\author{Soumen Manna}

\address{Indian Institute of Technology Bombay, India }

\begin{abstract}
Choice designs for the main effects model, broader main effects model and main effects plus specified interaction effects model are discussed in this paper. Universally optimal choice designs are obtained for all of these models using Hadamard matrix and other combinatorial techniques. Choice experiments under the multinomial logit model for equally attractive options are assumed for finding universally optimal choice designs.	
\end{abstract}

\begin{keyword}
Optimal choice designs \sep Hadamard matrix \sep Main effects \sep Interaction effects.
\end{keyword}

\end{frontmatter}

\section{Introduction}
Discrete choice experiments are widely used for quantifying consumer preferences in various areas including marketing, transport, environmental resource economics and public welfare analysis. A choice experiment consists of a number of choice sets, each containing several options (alternatives, profiles or treatment combinations). Respondents are shown each choice set in turn and are asked which option they prefer from each of the choice sets presented. Each option in a choice set is described by the level combination of $n$ factors (attributes). We assume that there are no repeated options in a choice set and each respondent chooses the best option from each choice set as per their perceived utility. A choice design is a collection of choice sets employed in a choice experiment. Thus a choice design $d$ consisting of $N$ choice sets, each containing $m$ profiles and each profile is a level combination of
$n$ factors.

\cite{r11} present a comprehensive exposition of designs for choice experiments under the multinomial logit model. More recently, \cite{r7} present a review of choice experiments till date. The literature so far on this subject is mainly focused on optimal choice designs for the main effects model  (\cite{r1}; \cite{r5}; \cite{r2}; \cite{r3} etc.) and for the main effects plus two factor interaction effects model (\cite{r1}; \cite{r4}; \cite{r2}; \cite{r6} etc.). Though for the former case, researchers are able to find optimal designs in very less number of choice sets but in the later case, the proposed optimal designs take a large number of choice sets. For this reason, in most of the cases, the theoretical optimal designs, for estimating all the main effects and two factor interaction effects, are quite impractical to use. To overcome this situation, researchers propose near-optimal designs with relatively lesser number of choice sets by sacrificing the efficiency (\cite{r9}; \cite{r10}; \cite{r11} etc.).      

Though the advancement of this subject is really splendid in the last decade but there many areas of investigation still remain to explore. In this paper we explore some of them which have some theoretical and practical importance. For example, suppose a researcher is interested to estimate all the main effects but (s)he can not deny the presence of two factor interaction effects in the model. So, in this situation, (s)he wants an optimal design for the estimation of main effects in the presence of two factor interaction effects but in the absence of three and higher order interaction effects in the model. We refer this model as broader main effects model and obtain optimal designs under this model. In most of the choice situations, the information of all the main effects and all the two factor interactions are not so important but the information about all the main effects and two or higher order interaction effects with some specific factors are much more  important to investigate. In this paper optimal choice designs are obtained for such situations in practical number of choice sets. It gives more flexibility to the researchers to design their own choice experiment appropriately for estimating all the main effects and specified two or higher order interaction effects of their interest.

In this paper we restrict ourselves in $2^n$ choice experiments under the multinomial logit model. A general set up and characterization of information matrix is given in Section 2. Optimal choice designs for the main effects under the main effects model and the broader main effects model are obtained in Section 3. In Section 4, optimal choice designs are obtained for the main effects and specified interaction effects (with one factor or more than one factors) model. Finally, In Section 5, a general discussion is given on achieved designs.    

\section{Useful notations and information matrix}
Let $d=d_{N,n,m} = \{(T_1, \ldots, T_m)\}$ is a choice design with $N$ choice sets each of size $m$ with $n$ factors, where a typical treatment combination $T_{i}=(i_{1} \ldots i_{r} \ldots i_{n}),i_{r}=0,1; r = 1,\ldots,n$. Let $A_i$, $i = 1,\ldots, m$, be $N\times n$ matrices with entries 0 and 1. Then a choice design $d$ can also be represented in matrix notation as $d=(A_1, \ldots, A_m)$, where the $p$-th row from each $A_i$ makes the $p$-th choice set  $S_{p}$ (say) and hence $d = \{ S_{p}: p = 1,\ldots,N\}$. 

Let $f_1, \ldots,f_n$, denote the $n$ factors and let $F_{h_1\ldots h_r}$ denotes the $r$-th order interaction effect corresponding to the factors $f_{h_1},\ldots,f_{h_r}$. Clearly, when $r=1$, $F_{h_1}$ denotes the main effect of the factor $f_{h_1}$ and when $r=2$, $F_{h_1h_2}$ denotes the two factor interaction effect between the factors $f_{h_1}$ and $f_{h_2}$ and so on. We define the position of the factor $f_{h_r}$ in a treatment $T_i$ as $i_{h_r}$ and the effective position of a factorial effect  
$F_{h_1\ldots h_r}$ in a treatment $T_i$ as $i^*_{h_1\ldots h_r} = r+1 - (i_{h_1}+ \cdots + i_{h_r})$ (mod 2). Let $S_{p}(h_1\ldots h_r)=(1^*_{h_1\ldots h_r}, \ldots, m^*_{h_1\ldots h_r})_{h_1\ldots h_r}$ be the effective choice set of $S_{p}=(T_1, \ldots, T_m)$ for the factorial effect $F_{h_1\ldots h_r}$. Similarly, let $S_{p}(h_1\ldots h_r,k_1\ldots k_l)=(1^*_{h_1\ldots h_r}1^*_{k_1\ldots k_l},\ldots,m^*_{h_1\ldots h_r}m^*_{k_1\ldots k_l})_{h_1\ldots h_r,k_1\ldots k_l}$ be the effective choice set of $S_{p}=(T_1, \ldots, T_m)$ for any two factorial effects $F_{h_1\ldots h_r}$ and $F_{k_1\ldots k_l}$.

In this paper we use the standard definition of contrast vector corresponding to a factorial effect $F_{h_1 \ldots h_r}$. Let $B_{h_u}=(b^{(j)}_{h_u})$ be the orthogonal contrast vector of the factorial effect $F_{h_u}$, $h_u=1,\ldots,n$. Corresponding to a treatment $T_i$ and the factorial effect $F_{h_u}$, let $b^{(i)}_{h_u}=-1$ if $i_{h_u}=0$ and $b^{(i)}_{h_u}=1$ if $i_{h_u}=1$. Let $B_{h_1 \ldots h_r}=(b^{(j)}_{h_1 \ldots h_r})$ be the orthogonal contrast vector of the factorial effect $F_{h_1 \ldots h_r}$. Then corresponding to a treatment $T_i$, $b^{(i)}_{h_1 \ldots h_r} = b^{(i)}_{h_1}\cdots b^{(i)}_{h_r}$. It is assumed that the treatments are arranged in lexicographic order in $B_{h_1 \ldots h_r}$. 

Let $\Lambda$ be the information matrix of treatment effects corresponding to a design $d$, and let $B$ be the orthogonal treatment contrast matrix corresponding to all the factorial effects of interest. Then the information matrix of the factorial effects of interest corresponding to $d$ is $C_d = (1/2^n)B\Lambda B'$ (see, street 2007 for details). A design $d$ is connected if the corresponding information matrix $C_d$ ($=C$ say) is positive definite. A  connected design allows the estimation of all underlying factorial effects of interest. 

Let $\mathcal{D}_{N,n,m}$ be the class of all connected designs with $N$ choice sets each of size $m$ with $n$ factors. Note from (street 2007) that for a design $d \in \mathcal{D}_{N,n,m} $, the $2^n\times 2^n$ information matrix $\Lambda = ((\lambda_{st}))$ of the treatment effects for equally attractive options is 

\begin{equation*}\label{}
\lambda_{st} = \left\{ \begin{array}{ccl}
((m-1)/N m^2) \sum_{j_2 < \cdots < j_m} N_{j_1 j_2 \ldots j_m} & \text{if} & s = t = j_1 \\

(-1/N m^2) \sum_{j_3 < \cdots < j_m} N_{j_1 j_2 \ldots j_m} &  \text{if} &  s =j_1, t = j_2 \\

0 &   \text{otherwise}, &
\end{array}\right.
\end{equation*}
where $N_{j_1 \ldots j_m}$ is the indicator function taking value 1 if $(T_{j_1}, \ldots, T_{j_m}) \in d$ and 0 otherwise. Let $M^{(j_{1} \ldots j_{m})} = ((m_{st}))$ be a $2^n\times 2^n$ matrix corresponding to a choice set $(T_{j_{1}},$  $\ldots, T_{j_{m}})$, where
\vspace{-.3cm}
$$ m_{st} = \left\{ \begin{array}{ll}
m-1 & \text{if} \hspace{.5cm} s = t, \;t \in \{ j_{1}, \ldots, j_{m} \} \\
-1 & \text{if} \hspace{.5cm} s\neq t, (s,t) \in \{j_{1}, \ldots, j_{m}\} \\
0 & \text{otherwise.} 
\end{array}\right.$$
Then for any choice design $d$ in $\mathcal{D}_{N,n,m}$, $\Lambda$ can be written as 
\begin{equation*} 
\Lambda = (1/N m^2)\sum_{j_{1}< \cdots <j_{m}} N_{j_1  \ldots j_m}M^{(j_1  \ldots j_m)} = (1/N m^2) \Lambda^* \hspace{.1cm} (say).
\end{equation*}
We consider the matrix $M^{(j_{1}\ldots j_{m})}$ as the contribution of the choice set $(T_{j_{1}},$  $\ldots, T_{j_{m}})$ to $\Lambda$.
The definition of $M^{(j_{1}\ldots j_{m})}$ suggests that we can write
$M^{(j_{1}\ldots j_{m})} = \sum_{j_{r}<j_{r'}}M^{(j_{r}j_{r'})}$,
where $j_{r},j_{r'} \in \{j_{1},\ldots ,j_{m}\}$. 
Thus the contribution of the choice set $(T_{j_{1}},$  $\ldots, T_{j_{m}})$ to $\Lambda$ is the sum of the contributions of all its $m(m-1)/2$ component pairs $(T_{j_{r}}, T_{j_{r'}})$. Then the information matrix $C=((c_{qq'}))$, for the factorial effects of interest can be written as  
\begin{eqnarray}\label{Cmat} C &=& (1/2^n)B\Lambda B' = (1/2^n N m^2) B\Lambda^* B' \nonumber \\
&=& (1/2^n N m^2) \sum_{j_{1}< \cdots <j_{m}} N_{j_1 \ldots j_m}\left\{B \left(\sum_{j_{r}<j_{r'}}M^{(j_{r}j_{r'})}\right) B'\right\}  .
\end{eqnarray}

\noindent Since each choice set $S_{p}$ contains $m(m-1)/2$ component pairs $(T_i,T_j)$ and there are $N$ such choice sets in $d$, therefore the total number of component pairs in a design $d$ is $N^* = N m(m-1)/2$. If $B_{h_1 \ldots h_r}$ and $B_{k_1 \ldots k_l}$ are the $q$-th and $q'$-th contrasts $B$, then from the expression (\ref{Cmat}) we see that the values of $c_{qq'}$ are only depend on the values of $B_{h_1 \ldots h_r}M^{(ij)}B'_{k_1 \ldots k_{l}}$, for all the $N^*$ component pairs $(T_i, T_j)$ in $d$. The following result helps us to determine the values of $c_{qq'}$. 
\begin{lemma}\label{lem-Mij}
	For a component pair $(T_i, T_j)$, the exhaustive cases indicating possible values of $B_{h_1 \ldots h_r}M^{(ij)}B'_{k_1 \ldots k_{l}}$ are \vspace{-.2cm} \\
	
	\noindent Case 1:
	$B_{h_1 \ldots h_r}M^{(ij)}B'_{k_1 \ldots k_l}=4$ when \\
	$(i^*_{h_1\ldots h_r}i^*_{k_1\ldots k_l}, j^*_{h_1\ldots h_r}j^*_{k_1\ldots k_l})_{h_1\ldots h_r, k_1\ldots k_l} = (00,11)_{h_1\ldots h_r, k_1\ldots k_l}$ \vspace{-.2cm} \\
	
	\noindent Case 2: 
	$B_{h_1 \ldots h_r}M^{(ij)}B'_{k_1 \ldots k_l}=-4$ when \\
	$(i^*_{h_1\ldots h_r}i^*_{k_1\ldots k_l}, j^*_{h_1\ldots h_r}j^*_{k_1\ldots k_l})_{h_1\ldots h_r, k_1\ldots k_l} = (01,10)_{h_1\ldots h_r, k_1\ldots k_l}$ \vspace{-.2cm} \\
	
	\noindent Case 3: $B_{h_1 \ldots h_r}M^{(ij)}B'_{k_1 \ldots k_l}=0$ for all other situations. 
\end{lemma}

\begin{proof}
	Let $B_{x} = (x_{1},\ldots, x_{i},\ldots, x_{j},\ldots,x_{2^n})$ and $B_{y} = (y_{1},\ldots, y_{i},\ldots, $ $y_{j},\ldots, y_{2^n})$ are the contrast vectors corresponding to $F_{h_1 \ldots h_r}$ and $F_{k_1 \ldots k_{l}}$ respectively.
	Note that $M^{(ij)}$ is a $ 2^n\times 2^n$ matrix with all elements $0$ except $ M_{ii}^{(ij)} =  M_{jj}^{(ij)} = 1$ and $ M_{ij}^{(ij)} =  M_{ji}^{(ij)} = -1$. Therefore $B_{x}M^{(ij)}B'_{y} = \left(0,\ldots,(x_{i}-x_{j}),\ldots, -(x_{i}-x_{j}),\ldots,0\right)B'_{y}$ \\
	$= (x_{i}-x_{j})y_{i} -(x_{i}-x_{j})y_j = (x_{i}-x_{j})(y_{i}-y_j)$.
	
	The results then simply follows from the definitions of $B_{h_1 \ldots h_r}$ and $B_{k_1 \ldots k_{l}}$.
\end{proof}

\begin{corollary}\label{cor-Mij}
	For a component pair $(T_i, T_j)$, $B_{h_1 \ldots h_r}M^{(ij)}B'_{h_1 \ldots h_{r}}=4$ when \\ $(i^*_{h_1\ldots h_r}, j^*_{h_1\ldots h_r})_{h_1\ldots h_r} = (0,1)_{h_1\ldots h_r}$, and 0 otherwise.
\end{corollary}

\begin{remark}\label{effectivePairs}
	From the Lemma \ref{lem-Mij} and Corollary \ref{cor-Mij}, we see that for a component pair $(T_i,T_j)$, the values of $B_{h_1 \ldots h_r}M^{(ij)}B'_{k_1 \ldots k_{l}}$ and $B_{h_1 \ldots h_r}M^{(ij)}B'_{h_1 \ldots h_{r}}$ are only depends on the effective component pairs $(T_i, T_j)_{h_1\ldots h_r, k_1\ldots k_l} = (i^*_{h_1\ldots h_r}i^*_{k_1\ldots k_l}, j^*_{h_1\ldots h_r}j^*_{k_1\ldots k_l})_{h_1\ldots h_r, k_1\ldots k_l}$ and  \\ $(T_i, T_j)_{h_1\ldots h_r} = (i^*_{h_1\ldots h_r}, j^*_{h_1\ldots h_r})_{h_1\ldots h_r}$ respectively.   
\end{remark}

Let $F$ is the set of $Q$ factorial effects $F_{h_1 \ldots h_r}$'s of interest corresponding to a choice design $d$. Then $B$ would be a $Q \times 2^n$ contrast matrix and $C$ would be a square matrix of order $Q$. Let $N_F = \{h_1 \ldots h_r : F_{h_1 \ldots h_r} \in F \}$. Now for any two $h_1 \ldots h_r, k_1 \ldots k_l \in N_F$, we define, $\eta_{h_1 \ldots h_r, k_1 \ldots k_l}^{+}$ and $\eta_{h_1 \ldots h_r, k_1 \ldots k_l}^{-}$ to be the total number of effective component pairs of the type $(00,11)_{h_1\ldots h_r, k_1\ldots k_l}$ and $(01,10)_{h_1\ldots h_r, k_1\ldots k_l}$ respectively in $d$. We also define, $n_{p}(h_1 \ldots h_r)$ to be the total number of 0's in the effective choice set $S_{p}(h_1 \ldots h_r)$ of $d$, $p=1,\ldots,N$, $h_1 \ldots h_r \in N_F$.    

We use the universal optimality criteria for finding optimal designs in $\mathcal{D}$. Following \cite{r8}, a choice design $d^*$ is universally optimal in $\mathcal{D}$, if $C_{d^*}$ is a scalar multiple of identity matrix and $ trace(C_{d^*}) \geq trace(C_{d})$, for any other design $d \in \mathcal{D}$. If a design $d$ is universally optimal in $\mathcal{D}$, then it is also $A$-, $D$-, and $E$-optimal. We have the following results for diagonal $C$-matrix and  maximum value of $trace(C)$ for a design $d$ in $\mathcal{D}_{N,n,m}$.  

\begin{lemma}\label{lem1}
	For $q\neq q'$, $c_{qq'}=0$, if and only if $\eta_{h_1 \ldots h_r, k_1 \ldots k_l}^{+} = \eta_{h_1 \ldots h_r, k_1 \ldots k_l}^{-}$.
\end{lemma}

\begin{proof}
	Let $c^*_{qq'}$ denotes the $(q,q')$-th element of $ B\Lambda^* B'$, $q \neq q'$. Then it follows from the expression (\ref{Cmat}) and from the Lemma \ref{lem-Mij} that
	\begin{eqnarray*}\label{Cstarxx}
		C^*_{qq'} &=& \sum_{j_{1}<\cdots <j_{m}} N_{j_1 \ldots j_m}\sum_{j_{r}<j_{r'}}\{B_{h_1 \ldots h_r}M^{(j_{r}j_{r'})} B'_{k_1 \ldots k_l}\}\\
		&=& \left[4(\eta_{h_1 \ldots h_r, k_1 \ldots k_l}^{+} - \eta_{h_1 \ldots h_r, k_1 \ldots k_l}^{-})+0\{N^* - (\eta_{h_1 \ldots h_r, k_1 \ldots k_l}^{+} + \eta_{h_1 \ldots h_r, k_1 \ldots k_l}^{-}) \}\right].
	\end{eqnarray*}
	Thus $c^*_{qq'}$ or equivalently $c_{qq'}=0$, if and only if $\eta_{h_1 \ldots h_r, k_1 \ldots k_l}^{+} = \eta_{h_1 \ldots h_r, k_1 \ldots k_l}^{-}$.
\end{proof}

\begin{lemma}\label{lem2}
	Let $d$ be a design in  $\mathcal{D}_{N,n,m}$, then
	$$ max(trace(C)) = \left\{ \begin{array}{cc}
	Q/2^n & \text{for $m$ even} \vspace{.2cm}\\
	Q(m^2-1)/2^nm^2 & \text{for $m$ odd,} \\
	\end{array}\right.$$
	and the $max(trace (C))$ occurs when $n_{p}(h_1 \ldots h_r) = m/2$ ($m$ even) and $n_{p}(h_1 \ldots h_r) = (m-1)/2$  or  $(m+1)/2$ ($m$ odd) for every effective choice set $S_{p}(h_1 \ldots h_r)$, $p=1,\ldots,N$, $h_1 \ldots h_r \in N_F$.
\end{lemma}

\begin{proof}
	Let $B_{h_1 \ldots h_r}$ be the $q$-th contrast of $B$ corresponding to the factorial effect $F_{h_1 \ldots h_r}$ and let $c^*_{qq}$ be the $(q,q)$-th element of  $B\Lambda^* B'$. Note from (\ref{Cmat}) that every component pair $(T_i,T_j)$ adds a value $B_{h_1 \ldots h_r} M^{(ij)}B'_{h_1 \ldots h_r}$ to $c^*_{qq}$. From Corollary \ref{cor-Mij} and Remark \ref{effectivePairs}, we see that this value is 4 if and only if the effective pair $(T_i, T_j)_{h_1\ldots h_r}$ has $i^*_{h_1 \ldots h_r} \neq j^*_{h_1 \ldots h_r}$. Since the contribution of each effective choice set $S_{p}(h_1 \ldots h_r)$ to $c^*_{qq}$ is equivalent to the contributions of all its $m(m-1)/2$ effective component pairs  $(T_i, T_j)_{h_1\ldots h_r}$, then each $S_{p}(h_1 \ldots h_r)$ adds a value $4n_{p}(h_1 \ldots h_r)(m-n_{p}(h_1 \ldots h_r))$ to $c^*_{qq}$. This value is maximum 
	when (i) $n_{p}(h_1 \ldots h_r)=m/2$ (for $m$ even) and (ii) $n_{p}(h_1 \ldots h_r)=(m-1)/2$ or $n_{p}(h_1 \ldots h_r)=(m+1)/2$ (for $m$ odd). Since there are $NQ$ such effective choice sets $S_{p}(h_1 \ldots h_r)$ in $d$, then we have the required expression for  $max(trace(C))$. 
\end{proof}

From Lemma \ref{lem1} and Lemma \ref{lem2}, it follows that a design $d \in \mathcal{D}_{N,n,m}$ is universally optimal for $F$, if (i) $C$ is diagonal, i.e., $\eta_{h_1 \ldots h_r, k_1 \ldots k_l}^{+} = \eta_{h_1 \ldots h_r, k_1 \ldots k_l}^{-}$, for all $h_1 \ldots h_r, k_1 \ldots k_l\in N_F$  and (ii) $ trace(C)$ is maximum, i.e.,  $n_{p}(h_1 \ldots h_r) = m/2$ ($m$ even) and 	$n_{p}(h_1 \ldots h_r) = (m-1)/2$  or  $(m+1)/2$ ($m$ odd), for all $S_{p}(h_1 \ldots h_r)$, $p=1,\ldots,N$, $h_1 \ldots h_r \in N_F$. Henceforth in this paper, by optimal design, we mean universally optimal choice design.  

\section{Optimal designs for the main effects and the broader main effects models} 
Let $F = \{F_{h_1}: h_1=1,\ldots,n \}$ be the set of all main effects of our interest. Let $F_{(2)} = \{F_{h_1h_2}: h_1 < h_2, h_1,h_2 =1,\ldots,n \}$ be the set of all two factor interaction effects and $N_{F_{(2)}}=\{h_1h_2: F_{h_1h_2} \in F_{(2)} \}$. Under the main effects model we obtain optimal designs for $F$ when two and higher order interactions effects are assumed to be zero and under the broader main effects we obtain optimal designs for $F$ when three and higher order interaction effects are assumed to be zero. Let $B_{(1)}$ be the contrast matrix corresponding to all the main effects and $B_{(2)}$ be the contrast matrix corresponding to all the two factor interaction effects. Then for a design $d \in \mathcal{D}_{N,n,m}$, the information matrix $C$ ($= C_{(1)}$ say) of $F$ under the main effects model is
$$ C_{(1)} = (1/2^n) B_{(1)} \Lambda B'_{(1)},$$ 
and the information matrix $C$ ($= C_{(2)}$ say) of $F$ under the broader main effects model is 
\begin{equation*}
C_{(2)} = (1/2^n) \{ B_{(1)} \Lambda B'_{(1)} - B_{(1)} \Lambda B'_{(2)}[B_{(2)} \Lambda B'_{(2)}]^{-}B_{(2)} \Lambda B'_{(1)}\}.
\end{equation*} 

Note that $B_{(1)} \Lambda B'_{(2)}[B_{(2)} \Lambda B'_{(2)}]^{-}B_{(2)} \Lambda B'_{(1)}$ is a non-negative definite matrix and 
$$trace(2^n C_{(2)}) = trace(B_{(1)} \Lambda B'_{(1)}) - trace(B_{(1)} \Lambda B'_{(2)}[B_{(2)} \Lambda B'_{(2)}]^{-}B_{(2)} \Lambda B'_{(1)}).$$
Thus $trace(C_{(2)}) \leq trace(C_{(1)})$ with equality attaining when $B_{(1)} \Lambda B'_{(2)}$ is a null matrix. Following Lemma \ref{lem1} and Lemma \ref{lem2}, we see that a design $d$ is optimal for estimating $F$ under the main effects model if (i) $\eta_{h_1, k_1}^{+} = \eta_{h_1, k_1}^{-}$, for all $h_1,k_1 \in N_F$ and (ii) 
$n_{p(h_1)} = m/2$ ( $m$ even) and 	$n_{p(h_1)} = (m-1)/2$  or  $(m+1)/2$ ($m$ odd), for all $S_{p}(h_1)$, $p=1,\ldots,N$, $h_1 \in N_F$. A design $d$ is optimal for estimating $F$ under the broader main effects model if it satisfies the above two conditions along with (iii) $\eta_{h_1, k_1k_2}^{+} = \eta_{h_1, k_1k_2}^{-}$, for all $h_1 \in N_F$, $k_1k_2 \in N_{F_{(2)}}$. Therefore if a design $d$ is optimal for estimating $F$ under the broader main effects model then it is also optimal under the main effects model but the converse is not always true. In what follows, in this section, we first obtain the optimal designs for $F$ under the broader main effects model and as a corollary we obtain the optimal design for $F$ under the main effects model. 

We use generator technique to construction optimal designs in this section and in the next section. Let $d = (A_1, \ldots, A_m)$ be a design in $\mathcal{D}_{N,n,m}$ and $g_j = (g_{j1},\ldots, g_{jn})$ be the $j$-th generator, where $g_{jr} = 0,1$. When $A_i$ is generated from $A_1$ using $g_j$, then it is denoted by $ A_i = A_1 + g_j$,
and is defined as $ A_i(p,r) =  A_1(p,r) + g_{jr} \hspace {0.2cm} (\text{mod 2})$, $1 \leq p \leq N , \hspace{.2cm}  1 \leq r \leq n$. Let $\bar{A_i}$ denotes the complement of $A_i$, i.e., the elements 0 and 1 interchange their respective positions in $A_i$ and $\bar{d} = (\bar{A}_1, \ldots, \bar{A}_m)$ denotes the complement design of $d = (A_1, \ldots, A_m)$. Similarly, let $\bar{T}_i$ and $\bar{g_i}$  are the complements of $T_i$ and $g_i$. We now obtain some optimal designs for given $n$ and $m$.

\begin{theorem}\label{th-Generator}
	Let $\nu$ and $\nu' (< \nu)$ are two conjugative numbers such that Hadamard matrices of order $\nu$ and $\nu'$ exist. Let $G = \{g_1, \ldots,  g_{\alpha}\}$ be a set of $\alpha$ different generators such that both $g_i$ and $\bar{g}_i \notin G$. Then for $m= even$, there exists a design $d^*_1$ in  $\mathcal{D}_{\nu, n, m}$, and for $m=odd$, there exists a design $d^*_2$ in  $\mathcal{D}_{2\nu, n, m}$, which are optimal for estimating $F$ under the broader main effects model, where $\nu' < n \leq \nu$, $m = 2,\ldots,2\alpha + 1,2\alpha + 2 $.
\end{theorem}

\begin{proof}
	Let $d = (A_1,  \ldots, A_m)$. Consider $H$ be a Hadamard matrix of order $\nu$. Let $A_1$ is derived from $H$ by taking any $n$ columns of $H$ and replacing all the $-1$ entries with 0 entries. Let $A_2 = \bar{A}_1$. Using the $\alpha$ generators from $G$, generate the other components of $d$ in the following manner \vspace{-.3cm} \\
	$$ A_{2u+1} = A_1 + g_u, A_{2u+2} = A_2 + g_u, \hspace{0.2 cm} u = 1, \ldots, \alpha. $$ 
	Consider, $d^*_1 = d$ ($m$ even) and $d^*_2 = \{\hspace{.1 cm} d,\bar{d}\hspace{.1 cm}\}$ ($m$ odd).  
	\begin{flushleft}
		{ Claim:} When $m= even$, $d^*_1$ is optimal in  $\mathcal{D}_{\nu, n, m}$ for estimating $F$ under the broader main effects model.
	\end{flushleft}
	Note that the information matrix $C$ of the design $d$ is the sum of information matrix corresponding to all component paired design $\delta_{ij} = (A_i, A_j)$, $i<j$, $i,j = 1,\ldots,m$. Note also that both $\delta_{ij}$ and  $\bar{\delta}_{ij}$ are present in $d^*_1$. Now according to the construction of $d^*_1$, we see that both the effective component pairs $(00, 11)_{h_1, k_1}$ and $(01, 10)_{h_1, k_1}$ occur equally often in each $\delta_{ij}$ of $d^*_1$. We also observe that the existence of the effective component pair $(00, 11)_{h_1, k_1k_2}$ in $\delta_{ij}$ ensures the existence of $(01, 10)_{h_1, k_1k_2}$ in the corresponding effective component pair of $\bar{\delta}_{ij}$ and vice-versa. Thus for the design $d^*_1$, $\eta_{h_1, k_1}^{+} = \eta_{h_1, k_1}^{-}$, for all $h_1,k_1 \in N_F$, and $\eta_{h_1, k_1k_2}^{+} = \eta_{h_1, k_1k_2}^{-}$, for all $h_1 \in N_F$, $k_1k_2 \in N_{F_{(2)}}$. Hence $C$ is a diagonal matrix with $B_{(1)} \Lambda B'_{(2)} = 0$. Also note that for every $S_{p}(h_1)$ of $d^*_1$, $n_p(h_1) = m/2$, $p = 1, \ldots, \nu$, $h_1 \in N_F$, and thus $trace(C)$ is maximum. Hence $d^*_1$ is optimal in $\mathcal{D}_{\nu,n,m}$ for estimating $F$ under the broader main effects model.

	\begin{flushleft}
		{ Claim:} When $m= odd$, $d^*_2$ is optimal in  $\mathcal{D}_{2\nu, n, m}$ for estimating $F$ under the broader main effects model.
	\end{flushleft}
	\vspace{-.1 cm}
	The proof follows in the same way as $m=even$ case on noting that for every component paired design $\delta_{ij}$ of $d$, the corresponding component paired design in $\bar{d}$ is $\bar{\delta}_{ij}$ and for every $S_{p}(h_1)$ of $d^*_2$, $n_p(h_1) = (m-1)/2$ or $(m+1)/2$, $p = 1, \ldots, 2\nu$, $h_1 \in N_F$. 
\end{proof}

\begin{corollary}
	For any odd or even $m$, both $d = (A_1, \ldots, A_m)$ and $\bar{d} = (\bar{A}_1,  \ldots, \bar{A}_m)$ are optimal in $\mathcal{D}_{\nu,n,m}$ for estimating $F$ under the main effects model.
\end{corollary}
\begin{exmp}
	Suppose $F = \{F_1,\ldots, F_8\}$ and we want to construct optimal designs for $m=5$ and $m=6$ under the broader main effects model. Take $H$ be the Hadamard matrix of order 8 and $A$ be the matrix is generated from $H$ by replacing all its $-1$ entries with 0 entries. For $m=6$, let $d_1 = (A_1,A_2,A_3,A_4,A_5,A_6)$, where $A_1= A$ and $ A_2 = \bar{A}$, $A_3=A_1+g_1$, $A_4=A_2+g_1$, $A_5=A_1+g_2$ and $A_6=A_2+g_2$. If we take $g_1 = (11100000)$ and $g_2 = (00000011)$, then
	\begin{center}
		$d_1 = \begin{array}{cccccc}
		(11111111,&00000000,&00011111,&11100000,&11111100,&00000011)\\
		(10101010,&01010101,&01001010,&10110101,&10101001,&01010110)\\
		(11001100,&00110011,&00101100,&11010011,&11001111,&00110000)\\
		(10011001,&01100110,&01111001,&10000110,&10011010,&01100101)\\
		(11110000,&00001111,&00010000,&11101111,&11110011,&00001100)\\
		(10100101,&01011010,&01000101,&10111010,&10100110,&01011001)\\
		(11000011,&00111100,&00100011,&11011100,&11000000,&00111111)\\
		(10010110,&01101001,&01110110,&10001001,&10010101,&01101010)
		\end{array}$
	\end{center}
	is optimal in $\mathcal{D}_{8,8,6}$ for $F$ under the broader main effects model.
	
	Now for $m=5$, let $d_2 = (A_1,A_2,A_3,A_4,A_5)$. Then $ d_{2}^* = \{d_2, \bar{d}_2\}$  is optimal in  $\mathcal{D}_{16,8,5}$ under the broader main effects model.
	Note that $d_1$ and $d_2$ are also optimal in $\mathcal{D}_{8,8,6}$ and $\mathcal{D}_{8,8,5}$ respectively for $F$ under the main effects model.  
\end{exmp}
The construction of Theorem \ref{th-Generator} is a general construction in a sense that for given $m$ and $n$, one can always find an optimal design for some $N$ $(\geq n)$, provided a Hadamard matrix of order $N$ exists. We now construct optimal designs for some specific values of $m$ and for any $n$, which are better than the designs provided by Theorem \ref{th-Generator} in a sense that they produce optimal designs in less number of choice sets ($N < n$, in most of the cases). For example, let $\nu$  be the least number greater than or equal to $n$ such that a Hadamard matrix of order $\nu$ exists. Let $H$ be a normalized Hadamard matrix of order $\nu$ and $A$ be the matrix is derived from $H$ by taking only $n$ columns of $H$ and replacing all the $-1$ entries with 0 entries. Let $A = (T_1, \ldots, T_{\nu})'$, where $T_i$ is a typical treatment combination. Then it is easy to see that
$ d^*_1 = (T_1, \ldots, T_{\nu}, \bar{T}_1,\ldots, \bar{T}_{\nu} )$ is optimal in $\mathcal{D}_{1,n,2\nu}$ for estimating $F$ under the broader main effects model. Now if the above $A$ does not contains the first column of $H$, then  
$d^*_2 = \{(T_1, \ldots, T_{\nu}), (\bar{T}_1,\ldots, \bar{T}_{\nu})\}$ is optimal in $\mathcal{D}_{2,n,\nu}$, $n < \nu$, for estimating $F$ under the broader main effects model. 
Note that the choice sets $(T_1, \ldots, T_{\nu})$ or $(\bar{T}_1,\ldots, \bar{T}_{\nu})$ is optimal in $\mathcal{D}_{1,n,\nu}$, $n < \nu$, for estimating $F$ under the main effects model.
\begin{exmp}\label{exmp-2}
	Let $H = (1111, 1010, 1100, 1001)'$ be a normalized Hadamard matrix of order 4, then 
	$d^*_1 = (1111, 1010, 1100, 1001, 0000, 0101, 0011, 0110)$
	is optimal in $\mathcal{D}_{1,4,8}$ for $F=\{F_1,F_2,F_3,F_4\}$ and $ d^*_2 = \{(111, 100, 010,$ $ 001), (000, 011, 101, 110 )\}$ is optimal in $\mathcal{D}_{2,3,4}$ for $F=\{F_1,F_2,F_3\}$ under the broader main effects model.
\end{exmp}

From the discussion above we see that for given $m=\nu$, the two-run design $d^*_2$ is optimal for $n <\nu$. We now generalize this idea for the cases when $n \geq \nu$. Before presenting the next construction, we need to define a new operation. Let $S_1 = (T_{11},\ldots,T_{1m})$ and $S_2 = (T_{21},\ldots,T_{2m})$ are two choice sets of size $m$ with $n_1$ and $n_2$ factors respectively. We denote the direct addition of $S_1$ and $S_2$ as $S_1 \oplus S_2$ and is defined by $S_1 \oplus S_2 = (T_{11}T_{21},\ldots,T_{1m}T_{2m})$, which is a new choice sets of size $m$ with $n_1+n_2$ factors. The definition carry forward for two choice designs in the similar way. Suppose $d_1 = \{S_{11},\ldots, S_{1N}\}$ and $d_2 = \{S_{21},\ldots, S_{2N}\}$ are two designs with $N$ choice sets with $n_1$ and $n_2$ factors respectively. Then the direct addition of $d_1$ and $d_2$ is denoted by $d_1 \oplus d_2$ and is defined by $d_1 \oplus d_2 = \{S_{11} \oplus S_{21},\ldots, S_{1N} \oplus S_{2N}\}$, which is a new design with $N$ choice sets with $n_1 + n_2$ factors.  

\begin{theorem}
	Let $H$ be a Hadamard matrix of order $\nu$. Then there exists an optimal design $d^*$ in $\mathcal{D}_{2^{\alpha+1},n,\nu}$ for estimating $F$ under the broader main effects model, $2^{\alpha-1}(\nu -1) < n \leq 2^{\alpha}(\nu -1)$, $\alpha = 1,2,3,\ldots$.
\end{theorem}

\begin{proof}
	Let $H$ be a normalized Hadamard matrix of order $\nu$ and $A$ be the $(\nu-1) \times \nu $ matrix is derived from $H$ by deleting the 1-st column and replacing all the $-1$ entries with $0$ entries. Let $d_0$ is the  choice set of size $\nu$ whose treatments are the rows of $A$. Then for $\alpha=1$, let \vspace{-.4cm}
	$$d_1 = \{d_0\oplus d_0, d_0\oplus \bar{d}_0\}.$$
	Similarly, for any $\alpha > 1$, \vspace{-.3cm}
	$$ d_{\alpha} = \{d_{\alpha-1}\oplus d_{\alpha-1},  d_{\alpha-1}\oplus \bar{d}_{\alpha-1}\}.$$
	Let $d$ be the design with $2^{\alpha}$ choice sets is generated from $d_{\alpha}$ by taking only $n$ corresponding factors from each treatment of each choice set of $d_{\alpha}$. Let $d^*= \{\hspace{.1cm} d,\bar{d} \hspace{.1cm}\}$. Thus $d^*$ is a design with $2^{\alpha+1}$ choice sets and $n$ factors, $2^{\alpha-1}(\nu -1) < n \leq 2^{\alpha} (\nu -1)$.  
	
	\begin{flushleft}
		{Claim:} $d^*$ is optimal in $\mathcal{D}_{2^{\alpha+1}, n, \nu}$ for estimating $F$ under the broader main effects model.
	\end{flushleft}
	Note from the construction of $d$ ( or $\bar{d}$ ) that both the effective component pairs $(00,11)_{h_1,k_1}$ and $(01,10)_{h_1,k_1}$ occur equally often in $d$ (or $\bar{d}$). Thus $\eta_{h_1, k_1}^{+} = \eta_{h_1, k_1}^{-}$, for all $h_1, k_1 \in N_F$.  We also observe that the existence of a effective component pair $(00, 11)_{h_1, k_1k_2}$ in $d$ ensures the existence of $(01, 10)_{h_1, k_1k_2}$ in the corresponding effective component pair of $\bar{d}$ and vice-versa. Thus $\eta_{h_1, k_1k_2}^{+} = \eta_{h_1, k_1k_2}^{-}$, for all $h_1 \in N_F$, $k_1k_2 \in N_{F_{(2)}}$. Also note that for every $S_{p}(h_1)$ of $d^*$, $n_p(h_1) = \nu /2$, $p=1,\ldots,2^{\alpha+1}$, $h_1 \in N_F$. Hence $d^*$ is optimal in $\mathcal{D}_{2^{\alpha+1}, n, \nu}$ for estimating $F$ under the broader main effects model. 
\end{proof}

\begin{corollary}
	Note that both $d$ and $\bar{d}$ are optimal in $\mathcal{D}_{2^{\alpha},n,\nu}$ for estimating $F$ under the  main effects model, $2^{\alpha-1}(\nu -1) < n \leq 2^{\alpha}(\nu -1)$, $\alpha = 1,2,3,\ldots$.
\end{corollary}

\begin{exmp}
	Suppose we want an optimal design for $F=\{F_1,F_2,F_3,F_4,F_5\}$ and $m =4$.\\ Let $H = (1111,  1100, 1010, 1001)'$ be a normalized Hadamard matrix of order 4, then 
	$ d_0 = \{111, 100, 010,001 \}$.
	Therefore for $\alpha =1$ \\
	$d_1 = \{d_0\oplus d_0, d_0\oplus \bar{d}_0\}= \{ ( 1 1 1 1 1 1, 1 0 0 1 0 0, 0 0 1 0 0 1, 0 1 0 0 1 0), (1 1 1 0 0 0, 1 0 0 0 1 1, 0 0 1 1 1 0, 0 1 0 1 0 1) \}$. Let $d$ be the design is obtained from $d_1$ removing the last factor of each treatments, then 
	$d^* = \{\hspace{.1cm} d, \bar{d} \hspace{.1cm} \} = $
	$\{ ( 1 1 1 1 1, 1 0 0 1 0, 0 0 1 0 0, 0 1 0 0 1), (1 1 1 0 0, 1 0 0 0 1, 0 0 1 1 1, 0 1 0 1 0), (0 0 0 0 0, 0 1 1 0 1, \\ 1 1 0 1 1, 1 0 1 1 0), (0 0 0 1 1, 0 1 1 1 0, 1 1 0 0 0, 1 0 1 0 1) \}$ is optimal in $\mathcal{D}_{4, 5, 4 }$ for $F$ under the broader main effects model. Also note that  
	$d = \{ ( 1 1 1 1 1 , 1 0 0 1 0, 0 0 1 0 0, 0 1 0 0 1), (1 1 1 0 0, 1 0 0 0 1, 0 0 1 1 1, 0 1 0 1 0 )\}$ or $ \bar{d} = \{ 
	(0 0 0 0 0, 0 1 1 0 1, 1 1 0 1 1, 1 0 1 1 0), (0 0 0 1 1, 0 1 1 1 0, 1 1 0 0 0, 1 0 1 0 1)\} $ are optimal in $\mathcal{D}_{2, 5, 4 }$ for $F$ under the main effects model. 
\end{exmp}

\section{Optimal designs for the main effects and the specified interaction effects model}
In many choice investigation problems researchers need the information about two and higher order interaction effects along with the main effects. In most of such cases, the information about all the interaction effects are not so important but interaction effects with some specified factors are much more important. In this section we obtain optimal designs for such situations. We first present constructions of optimal designs for estimating all the main effects and all the interaction effects with a single specified factor and later on we generalize this idea with more than one specified factors. Without loss of generality we assume the first factor to be the specified factor. Thus if $F$ is the set of all factorial effects of our interest, then  $F = \{F_1, \ldots, F_n$, $F_{12}, \ldots, F_{1n}, F_{123}, \ldots, F_{12 \ldots n}\}$. The factorial effects which are not in $F$ are assumed to be zero in this section. For a component matrix $A_i$ of a design $d = (A_1, \ldots, A_i, \ldots, A_m)$, we define ${\mathcal C}_{A_i}(h_1 \ldots h_r) = \left( i^*_{h_1 \ldots h_r} \right)_{N \times 1}$ to be the effective column of $A_i$ corresponding to the factorial effect $F_{h_1 \ldots h_r}$.  Let $H$ is a Hadamard matrix of order $2^{\alpha}$, $\alpha \geq 2$, where 
\begin{equation} \label{Hadmard-H}
H = \left[ \begin{array}{rr} 1 & 1 \\ 1 & -1 \end{array} \right] \otimes \cdots \otimes \left[\begin{array}{rr} 1 & 1 \\ 1 & -1 \end{array} \right].
\end{equation}

\noindent Here $\otimes$ denotes the Kronecker product. Let 
\begin{equation}\label{matrix-A}
A =  \{\text{a (0,1) matrix is derived from H of (\ref{Hadmard-H}) after replacing all the $-1$ by $0$}\}.
\end{equation}

\noindent  Note that every effective column ${\mathcal C}_{A}(h_1 \ldots h_r)$ is equivalent to a column of $A$. We use this matrix $A$ in many of our constructions of optimal designs in this section.  

\begin{theorem}\label{th-speAll}
	For $2^{\alpha - 1} <n \leq 2^{\alpha}$, $\alpha \geq 2$, there exists an optimal design $d^*$ in $\mathcal{D}_{2^{\alpha}, n, 4}$ for estimating $F$.
\end{theorem}

\begin{proof}
	Let $A_1$ is a $2^{\alpha} \times n$ matrix is derived from (\ref{matrix-A}) by taking any $n$ columns of $A$ (including the first column), $2^{\alpha - 1} <n \leq 2^{\alpha}$, $\alpha \geq 2$. Consider  $d^* = (A_1, A_2, A_3, A_4)$, where $A_2 = \bar{A}_1$, $A_3 = A_1 + g$ and $A_4 = A_2 + g$,  with $g = (100...0)$.
	\begin{flushleft}
		Claim: $d^*$ is optimal in $\mathcal{D}_{2^{\alpha}, n, 4}$ for estimating $F$. 
	\end{flushleft}   
	Note that every effective column ${\mathcal C}_{A_1}(h_1 \ldots h_r)$ changes twice in $d^*$.  Therefore $n_{p}(h_1 \ldots h_r) = 2$, for every $S_{p}(h_1 \ldots h_r)$ in $d^*$, $p = 1,\ldots, 2^{\alpha}$, $h_1 \ldots h_r \in N_F$, and hence $trace(C)$ is maximum. 
	
	Now we only need to show that $C$ is diagonal. Let $w_1 = i^*_{h_1 \ldots h_r}$ and $w_2 = i^*_{k_1 \ldots k_l}$ are the $p$-th elements of ${\mathcal C}_{A_1}(h_1 \ldots h_r)$ and ${\mathcal C}_{A_1}(k_1 \ldots k_l)$ respectively. Since every effective column of $A_1$ changes exactly two times in $d^*$, then each effective choice set $S_{p}(h_1 \ldots h_r, k_1 \ldots k_l)$ of $d^*$ is any of the following two types
	\begin{description}
		\item[\it type-1:] $S_{p}(h_1 \ldots h_r, k_1 \ldots k_l) \equiv (w_1w_2, w_1w_2, \bar{w}_1\bar{w}_2, \bar{w}_1\bar{w}_2 )_{h_1 \ldots h_r, k_1 \ldots k_l}$
		\item[\it type-2:] $S_{p}(h_1 \ldots h_r, k_1 \ldots k_l) \equiv (w_1w_2, \bar{w}_1\bar{w}_2, w_1\bar{w}_2, \bar{w}_1w_2)_{h_1 \ldots h_r, k_1 \ldots k_l}$ 
	\end{description}
	Note that if $S_{p}(h_1 \ldots h_r, k_1 \ldots k_l)$ is of {\it type-1}, then $(00,11)_{h_1 \ldots h_r, k_1 \ldots k_l}$ occurs twice when $w_1=w_2$ and $(01,10)_{h_1 \ldots h_r, k_1 \ldots k_l}$ occurs twice when $w_1\neq w_2$ in $S_{p}(h_1 \ldots h_r, k_1 \ldots k_l)$. Similarly, if $S_{p}(h_1 \ldots h_r, k_1 \ldots k_l)$ is of {\it type-2}, then both $(00,11)_{h_1 \ldots h_r, k_1 \ldots k_l}$ and $(01,10)_{h_1 \ldots h_r, k_1 \ldots k_l}$ occur once each in $S_{p}(h_1 \ldots h_r, k_1 \ldots k_l)$. Now we have the following two cases

	\begin{description}
		\item[\it Case-1:] If  ${\mathcal C}_{A_1}(h_1 \ldots h_r)$ and ${\mathcal C}_{A_1}(k_1 \ldots k_l)$ are different, then there are equal number of effective choice sets $S_{p}(h_1 \ldots h_r, k_1 \ldots k_l)$ in $d^*$ for which $w_1 = w_2$ and $w_1 \neq w_2$. Therefore whether it is of {\it type-1} or {\it type-2}, $\eta_{h_1 \ldots h_r, k_1 \ldots k_l}^{+} = \eta_{h_1 \ldots h_r, k_1 \ldots k_l}^{-}$, for all $h_1 \ldots h_r, k_1 \ldots k_l \in N_F$.  
		
		\item [\it Case-2:] If ${\mathcal C}_{A_1}(h_1 \ldots h_r)$ and ${\mathcal C}_{A_1}(k_1 \ldots k_l)$ are same, then $S_{p}(h_1 \ldots h_r, k_1 \ldots k_l)$ is of {\it type-2}. Therefore  $\eta_{h_1 \ldots h_r, k_1 \ldots k_l}^{+} = \eta_{h_1 \ldots h_r, k_1 \ldots k_l}^{-}$, for all $h_1 \ldots h_r, k_1 \ldots k_l \in N_F$.  
	\end{description}
	Thus $C$ is a equal diagonal matrix with maximum $trace$ and hence $d^*$ is optimal in $\mathcal{D}_{2^{\alpha}, n, 4}$ for estimating $F$.   
\end{proof}

\begin{corollary}\label{cor-genAll}
	Let $F^* \subset F$, then $d^*$ is also optimal in  $\mathcal{D}_{2^{\alpha}, n, 4}$ for estimating $F^*$. 
\end{corollary}

\begin{theorem}\label{th-speAll-3}
	For $2^{\alpha - 1} <n \leq 2^{\alpha}$, $\alpha \geq 2$, there exists an optimal design $d^*$ in $\mathcal{D}_{2^{\alpha+1}, n, 3}$ for estimating $F$.
\end{theorem}

\begin{proof}
	Let $A_1$ is a $2^{\alpha} \times n$ matrix is derived from (\ref{matrix-A}) by taking any $n$ columns of $A$ (including the first column), $2^{\alpha - 1} <n \leq 2^{\alpha}$. Let $d = (A_1, A_2, A_3)$, where $A_2 = \bar{A}_1$, $A_3 = A_1 + g$ with $g = (100...0)$. Consider $d^* =\{\hspace{.1 cm} d, \bar{d} \hspace{.1 cm} \}$. 
	\begin{flushleft}
		Claim: $d^*$ is optimal in $\mathcal{D}_{2^{\alpha +1}, n, 3}$ for estimating $F$. 
	\end{flushleft}   
	
	Note that every effective column ${\mathcal C}_{A_1}(h_1 \ldots h_r)$ changes either once or twice in $d^*$. Therefore $n_{p}(h_1 \ldots h_r) = 1$ or $n_{p}(h_1 \ldots h_r) = 2$, for every $S_{p}(h_1 \ldots h_r)$ in $d^*$, $p = 1,\ldots, 2^{\alpha+1}$, $h_1 \ldots h_r \in N_F$, and hence $trace(C)$ is maximum. 
	
	Now we only need to show that $C$ is diagonal. Let $w_1 = i^*_{h_1 \ldots h_r}$ and $w_2 = i^*_{k_1 \ldots k_l}$ are the $p$-th elements of ${\mathcal C}_{A_1}(h_1 \ldots h_r)$ and ${\mathcal C}_{A_1}(k_1 \ldots k_l)$ respectively. Since every effective column of $A_1$ changes either once or twice in $d^*$, then  each effective choice set $S_{p}(h_1 \ldots h_r, k_1 \ldots k_l)$ of $d^*$ is any of the following types
	\begin{description}
		\item When both ${\mathcal C}_{A_1}(h_1 \ldots h_r)$ and ${\mathcal C}_{A_1}(k_1 \ldots k_l)$ change once 
		
		{\it type-1:} $ S_{p}(h_1 \ldots h_r, k_1 \ldots k_l) \equiv (w_1w_2, w_1w_2, \bar{w}_1\bar{w}_2)_{h_1 \ldots h_r, k_1 \ldots k_l}$
		
		{\it type-2:} $S_{p}(h_1 \ldots h_r, k_1 \ldots k_l) \equiv (w_1w_2, w_1\bar{w}_2, \bar{w}_1w_2)_{h_1 \ldots h_r, k_1 \ldots k_l}$ 
		
		\item When ${\mathcal C}_{A_1}(h_1 \ldots h_r)$  changes once and ${\mathcal C}_{A_1}(k_1 \ldots k_l)$  changes twice or vice-versa 
		
		{\it type-3:} $ S_{p}(h_1 \ldots h_r, k_1 \ldots k_l) \equiv (w_1w_2, w_1\bar{w}_2, \bar{w}_1\bar{w}_2)_{h_1 \ldots h_r, k_1 \ldots k_l}$
		
		{\it type-4:} $S_{p}(h_1 \ldots h_r, k_1 \ldots k_l) \equiv (w_1w_2, \bar{w}_1w_2, \bar{w}_1\bar{w}_2)_{h_1 \ldots h_r, k_1 \ldots k_l}$ 
		
		\item When both ${\mathcal C}_{A_1}(h_1 \ldots h_r)$ and ${\mathcal C}_{A_1}(k_1 \ldots k_l)$  change twice
		
		{\it type-5:} $ S_{p}(h_1 \ldots h_r, k_1 \ldots k_l) \equiv (w_1w_2, \bar{w}_1\bar{w}_2, \bar{w}_1\bar{w}_2)_{h_1 \ldots h_r, k_1 \ldots k_l}$
		
	\end{description}
	Now we have the following two cases.
	
	\begin{description}
		\item[\it Case-1:] If ${\mathcal C}^{h_1 \ldots h_r}_{A_1}$ and ${\mathcal C}^{k_1 \ldots k_l}_{A_1}$ are different, then there are equal number of effective choice sets  $S_{p}(h_1 \ldots h_r, k_1 \ldots k_l)$ in $d^*$, for which $w_1 = w_2$ and $w_1 \neq w_2$. From any of the above five types we see that if an effective choice set with $w_1 = w_2$ contains $(00,11)_{h_1 \ldots h_r, k_1 \ldots k_l}$,  then an effective choice set with $w_1 \neq w_2$ contains $(01,10)_{h_1 \ldots h_r, k_1 \ldots k_l}$.    
		
		\item [\it Case-2:] If ${\mathcal C}^{h_1 \ldots h_r}_{A_1}$ and ${\mathcal C}^{k_1 \ldots k_l}_{A_1}$ are same, then every effective choice set $S_{p}(h_1 \ldots h_r, k_1 \ldots k_l)$ of $d^*$ is either of {\it type-2} or {\it type-3} or {\it type-4}. Note that if an effective choice set in $d$ is of {\it type-2}, then the corresponding choice set of $\bar{d}$  is of {\it type-3} or {\it type-4} and vice-versa. Thus if an effective choice set of $d$ contains $(00,11)_{h_1 \ldots h_r, k_1 \ldots k_l}$, then the corresponding effective choice set of $\bar{d}$ contains $(01,10)_{h_1 \ldots h_r, k_1 \ldots k_l}$ and vice-versa.   
	\end{description}
	
	Considering both the above cases we see that $\eta_{h_1 \ldots h_r, k_1 \ldots k_l}^{+} = \eta_{h_1 \ldots h_r, k_1 \ldots k_l}^{-}$, for all $h_1 \ldots h_r$, $k_1 \ldots k_l \in N_F$. Thus $C$ is a equal diagonal matrix with maximum $trace$ and hence $d^*$ is optimal in $\mathcal{D}_{2^{\alpha+1}, n, 3}$ for estimating $F$.   
\end{proof}

\begin{corollary}\label{cor-genAll-3}
	Let $F^* \subset F$, then $d^*$ is also optimal in  $\mathcal{D}_{2^{\alpha+1}, n, 3}$ for estimating $F^*$. 
\end{corollary}

The results of Corollary \ref{cor-genAll} and Corollary \ref{cor-genAll-3} can be further improved in terms of lesser number of choice sets if $F^*$ contains all the main effects and all the specified two factor interaction effects with one factor, i.e., $F^* = \{F_1,\ldots,F_n,F_1F_2,\ldots,F_1F_n \}$.  

\begin{theorem}
	For given $n$, let $\nu$ be the least number such that a Hadamard matrix of order $\nu$ exists. Then there exists an optimal design $d^*$ in $\mathcal{D}_{\nu, n, 4}$ for estimating $F^*$.   
\end{theorem}

\begin{proof}
	The proof follows same way as of Theorem \ref{th-speAll}, by taking  $H$ of (\ref{Hadmard-H}) to be a normalized Hadamard matrix of order $\nu$. 
\end{proof}

\begin{theorem}
	For given $n$, let $\nu$ be the least number such that a Hadamard matrix of order $\nu$ exists. Then there exists an optimal design $d^*$ in $\mathcal{D}_{2\nu, n, 3}$ for estimating $F^*$.   
\end{theorem}

\begin{proof}
	The proof follows same way as of Theorem \ref{th-speAll-3}, by taking  $H$ of (\ref{Hadmard-H}) to be a normalized Hadamard matrix of order $\nu$. 
\end{proof}

\begin{exmp}
	Suppose we have $n=4$ factors and we want an optimal design for $F = \{ F_1, F_2, F_3, F_4, F_{12}, F_{13}, F_{14}, F_{123}, F_{134}, F_{1234}\}$ or any subset $F^*$ of $F$. Then 
	\begin{center}
		$d^* = \begin{array}{cccc}
		(1111,&0000,&0111,&1000)\\
		(1010,&0101,&0010,&1101)\\
		(1100,&0011,&0100,&1011)\\
		(1001,&0110,&0001,&1110)
		\end{array}$
	\end{center}
	is optimal in $\mathcal{D}_{4, 4, 4}$ for estimating $F$ or any subset $F^*$ of $F$.
\end{exmp}

So far we discuss about the optimal designs for estimating main effects and specified interaction effects with one factor. We now generalized the idea to get optimal designs for estimating main effects and specified interaction effects with more than one factor. Suppose we have a partition of factors into two groups, say, $\mathbf{f_1} = \{ f_{h_1}, \ldots, f_{h_r}\}$ and $\mathbf{f_2} = \{ f_{k_1}, \ldots, f_{k_l}\}$ and we are interested in estimating all the main effects and all the specific interaction effects between each factor of the group $\mathbf{f_1}$ to all the factors of the group $\mathbf{f_2}$. Therefore the set of all factorial effects of interests are $F = \{F_{h_1}, \ldots, F_{h_r}$, $F_{k_1}, \ldots, F_{k_l}$, $F_{h_1k_1},\ldots, F_{h_1k_1\ldots k_l}$, $ F_{h_2k_1},\ldots, F_{h_rk_1\ldots k_l}\}$.   

\begin{theorem}\label{th-speAll1}
	for $2^{\alpha - 1} <n \leq 2^{\alpha}$, $\alpha \geq 2$, there exists an optimal design $d^*$ in $\mathcal{D}_{2^{\alpha}, n, 4}$ for estimating $F$.
\end{theorem}

\begin{proof}
	Let $A_1$ is a $2^{\alpha} \times n$ matrix is derived from (\ref{matrix-A}) by taking any $n$ columns of $A$ (including the first column), $2^{\alpha - 1} <n \leq 2^{\alpha}$, $\alpha \geq 2$. Consider  $d^* = (A_1, A_2, A_3, A_4)$, where $A_2 = \bar{A}_1$, $A_3 = A_1 + g$ and $A_4 = A_2 + g$. Here $g$ is a generator whose first $r$ elements are 1, i.e., $g = (11...10...0)$. The rest of the proof follows in the similar way as Theorem \ref{th-speAll}.
\end{proof}

\begin{corollary}
	Let $F^* \subset F$, then $d^*$ is also optimal in  $\mathcal{D}_{2^{\alpha}, n, 4}$ for estimating $F^*$. 
\end{corollary}

\begin{theorem}\label{th-speAll2}
	for $2^{\alpha - 1} <n \leq 2^{\alpha}$, $\alpha \geq 2$, there exists an optimal design $d^*$ in $\mathcal{D}_{2^{\alpha+1}, n, 3}$ for estimating $F$.
\end{theorem}

\begin{proof}
	Let $A_1$ is a $2^{\alpha} \times n$ matrix is derived from (\ref{matrix-A}) by taking any $n$ columns of $A$ (including the first column), $2^{\alpha - 1} <n \leq 2^{\alpha}$. Let $d = (A_1, A_2, A_3)$, where $A_2 = \bar{A}_1$, $A_3 = A_1 + g$. Here $g$ is a generator whose first $r$ elements are 1, i.e., $g = (11...10...0)$. Consider $d^* =\{d, \bar{d}\}$. 
	The rest of the proof follows same way as Theorem \ref{th-speAll-3}.
\end{proof}

\begin{corollary}
	Let $F^* \subset F$, then $d^*$ is also optimal in  $\mathcal{D}_{2^{\alpha+1}, n, 3}$ for estimating $F^*$. 
\end{corollary}

\begin{exmp}
	Suppose we have $n=4$ factors and we want an optimal design for $F = \{ F_1, F_2, F_3, F_4, F_{13}, F_{14}, F_{23}, F_{24}, F_{134}, F_{234}\}$ or any subset $F^*$ of $F$. Then 
	\begin{center}
		$d^* = \begin{array}{cccc}
		(1111,&0000,&0011,&1100)\\
		(1010,&0101,&0110,&1001)\\
		(1100,&0011,&0000,&1111)\\
		(1001,&0110,&0101,&1010)
		\end{array}$
	\end{center}
	is optimal in $\mathcal{D}_{4, 4, 4}$ for estimating $F$ or any subset $F^*$ of $F$.
\end{exmp}

\begin{table}[ht]
	\caption{Number of choice sets ($N$) required  for the optimal designs for given $m$ and $n$.}
	\centering
	\begin{tabular}{|c|ccccccccccc|}
		\hline
		\multicolumn{12}{|c|}{Main effects model} \\
		\hline
		m $\diagdown$ n	& 2 &3 &4 &	5 &	6 & 7 &	8 & 9 &10 &11 &12 \\
		\hline
		2 &	2 &	4 &	4 &8 &8 &8 &8 &12 &12 &12 &12  \\
		3 &2 &4 &4 &8 &8 &8 &8 &12 &12 &12 &12 \\
		4 &1 &1 &2 &2 &2 &4 &4 &4 &4 &4 &4 \\ 
		5 & &4 &4 &8 &8 &8 &8 &12 &12 &12 &12 \\
		6 & &4 &4 &8 &8 &8 &8 &12 &12 &12 &12 \\
		7 & &4 &4 &8 &8 &8 &8 &12 &12 &12 &12 \\
		8 & &1 &1 &1 &1 &1 &2 &2 &2 &2 &2 \\	
		\hline
		\multicolumn{12}{|c|}{Broader main effects model} \\
		\hline
		2 &2 &4 &4 &8 &8 &8 &8 &12 &12 &12 &12 \\
		3 &2 &8 &8 &16 &16 &16 &16 &24 &24 &24 &24 \\
		4 &1 &2 &4 &4 &4 &8 &8 &8 &8 &8 &8 \\
		5 & &8 &8 &16 &16 &16 &16 &24 &24 &24 &24 \\
		6 & &4 &4 &8 &8 &8 &8 &12 &12 &12 &12 \\
		7 & &8 &8 &16 &16 &16 &16 &24 &24 &24 &24 \\
		8 & &1 &1 &2 &2 &2 &4 &4 &4 &4 &4 \\
		\hline
		\multicolumn{12}{|c|}{Main plus specified two factor interaction effects model } \\
		\hline
		3 &4 &8 &8 &16 &16 &16 &16 &24 &24 &24 &24 \\
		4 & &4 &4 &8 &8 &8 &8 &12 &12 &12 &12 \\
		\hline
		\multicolumn{12}{|c|}{Main plus all specified interaction effects model } \\
		\hline
		3 &4 &8 &8 &16 &16 &16 &16 &32 &32 &32 &32 \\
		4 & &4 &4 &8 &8 &8 &8 &16 &16 &16 &16 \\
		\hline
	\end{tabular}	
	\label{tb-1}
\end{table}

\section{Discussion}
In this paper, we have obtained optimal choice designs for some broader class of model set up than the existing ones in the literature. The constructions are very easy and simple, yet the optimal designs are obtained in practical number of choice sets. From the Table \ref{tb-1}, it is seen that when $m=4t$, $t=1,2,\ldots$, one gets optimal designs in least number of choice sets than any other $m$ for each of the model set up. When $m=4t+2$, $t=0,1,\ldots$, optimal designs for the broader main effects model are obtained in same number of choice sets as the main effects model but when $m$ is odd it takes double. 

Considering the fact that all the two or higher order interactions effects are not equally important in any choice investigation problem, the designs present in this paper are quite useful to reduce the cognitive burden of the respondents. For example, in a $2^4$ choice investigation problem, for an optimal design, one needs 80 choice sets of size 2 to estimate all the main effects and all the two factor interaction effects (\cite{r9}), whereas one needs only 4 choice sets of size 4 to estimate all the main effects and all the two and higher order specified interaction effects.


%
%


\newpage 
{\bf References}
\begin{thebibliography}{00}
	\bibitem[Burgess and Street(2003)]{r1} Burgess, L., Street, D. J., 2003. Optimal designs for $2^k$ choice experiments. Commun. stat.- Theory and Methods 32, 2185-2206.
	\vspace{-.2cm}
	
	\bibitem[Burgess and Street(2005)]{r2} Burgess, L., Street, D. J., 2005. Optimal designs for asymmetric choice experiments. J. Statist. Inference 134, 228-301.
	\vspace{-.2cm}
	
	\bibitem[Demirkale et al.(2013)]{r3} Demirkale, F., Donovan, D., Street, D. J., 2013. Constructing D-optimal symmetric stated preference discrete choice experiments. J. Statist. Plann. Inference 143, 1380-1391.
	\vspace{-.2cm}
	
	\bibitem [Gra\ss hoff et al.(2003)]{r4} Gra\ss hoff, U., Gro\ss mann, H., Holling, H., Schwabe, R., 2003. Optimal paired comparison designs for first-order interactions. Statistics 37, 373-386.
	\vspace{-.2cm}
	
	\bibitem[Gra\ss hoff et al.(2004)]{r5} Gra\ss hoff, U., Gro\ss mann, H., Holling, H., Schwabe, R., 2004. Optimal designs for main effects in linear paired comparison models. J. Statist. Plann. Inference 126, 361-376.
	\vspace{-.2cm}
	
	\bibitem[Gro\ss mann et al.(2012)]{r6} Gro\ss mann, H., Schwabe, R., Gilmour, S.G., 2012. Designs for first-order interactions in paired comparison experiments with two-level factors. J. Statist. Plann. Inference 142, 2395-2401. 
	\vspace{-.2cm}
	
	\bibitem[Gro\ss mann and Schwabe(2015)]{r7} Gro\ss mann, H., Schwabe, R., 2015. Design for discrete choice experiments, in: Dean, A., Morris, M., Stufken, J., Bingham, D. (Eds.), Handbook of Design and Analysis of Experiments. Chapman \& Hall/CRC press, Boca Raton, pp. 791-835.
	\vspace{-.2cm}
	
	\bibitem[Kiefer(1975)]{r8} Kiefer, J., 1975. Construction and optimality of generalized Youden designs, in: Srivastava, J.N. (Eds.), A Survey of Statistical Design and Linear Models. North-Holland, Amsterdam, pp. 333-353.
	\vspace{-.2cm}
	\bibitem[Street and Burgess(2004)]{r9} Street, D. J., Burgess, L., 2004. Optimal and near-optimal pairs for the estimation of effects in 2-level choice experiments. J. Statist. Plann. Inference 118, 185-199.
	\vspace{-.2cm}
	
	\bibitem[Street et al.(2005)]{r10} Street, D. J., Burgess, L., Louviere, J.J., 2005. Quick and easy choice sets: Constructing optimal and nearly optimal stated choice experiments. Intern. J. of Research in Marketing 22, 459-470.
	
	\vspace{-.2cm}
	\bibitem[Street and Burgess(2007)]{r11} Street, D. J., Burgess, L., 2007. The Construction of Optimal Stated Choice Experiments: Theory and Methods. Wiley, New York. 
	
\end{thebibliography}
\end{document}